\documentclass[lettersize,journal]{IEEEtran}
\usepackage{amsmath,amsfonts}
\usepackage{algorithmic}
\usepackage{algorithm}
\usepackage{array}
\usepackage[caption=false,font=normalsize,labelfont=sf,textfont=sf]{subfig}
\usepackage{textcomp}
\usepackage{stfloats}
\usepackage{url}
\usepackage{verbatim}
\usepackage{graphicx}
\usepackage{cite}
\usepackage{pstricks}
\usepackage{auto-pst-pdf}
\usepackage{color}
\usepackage{amsthm}
\usepackage{amssymb}
\usepackage{booktabs}
\newtheorem{proposition}{Proposition}
   \newtheorem{example}{Example}
   
      \newtheorem{observation}{Observation}

\begin{document}

\title{Introducing Modelling, Analysis and Control\\of Three-Phase Electrical Systems  \\Using  Geometric  Algebra}

\author{Manel Velasco, 
            Isiah Zaplana,
            Arnau Dòria-Cerezo, 
            Josué Duarte and Pau Martí
\thanks{M.~Velasco and P.~Mart\'{i} 
 are with the Automatic Control Department, Universitat Politècnica de Catalunya,  Barcelona, Spain, 
email: \{manel.velasco, josue.duarte, pau.marti\}@upc.edu}
\thanks{J.~Duarte is with the Electrical Engineering, Escola Universit\`{a}ria Salesiana de Sarri\`{a}, 08017 Barcelona, Spain, 
email: jnduarte@euss.cat}
\thanks{I.~Zaplana is with the  Automatic Control Department and Inst. of Industrial and Control Engineering, Universitat Politècnica de Catalunya,  Barcelona, Spain, 
email: isiah.zaplana@upc.edu}
\thanks{A.~Dòria-Cerezo is with the Dept. of Electrical Engineering and Inst. of Industrial and Control Engineering, Universitat Politècnica de Catalunya,  Barcelona, Spain, 
email: arnau.doria@upc.edu}
\thanks{Manuscript received April 19, 2021; revised August 16, 2021.}}

\markboth{Journal of \LaTeX\ Class Files,~Vol.~14, No.~8, August~2021}%
{Shell \MakeLowercase{\textit{et al.}}: A Sample Article Using IEEEtran.cls for IEEE Journals}


\maketitle

\begin{abstract}

This paper introduces a novel framework based on geometric algebra (GA) for the modeling, analysis, and control of three-phase electrical systems, with a focus on unbalanced scenarios.
The proposed approach utilizes GA to represent these systems using simplified single-input/single-output (SISO) models, significantly reducing their mathematical complexity while maintaining accuracy.
The framework introduces GA-valued transfer functions (GA-TFs), enabling the extension of well-established multivariable (MIMO) techniques, such as Youla-Kučera parameterization, into the GA domain as an illustrative example. This demonstrates the adaptability of classical tools within the new algebra.
The validity and potential of the approach are showcased through numerical examples of balanced and unbalanced systems, as well as physical implementations in a laboratory setup.
The results confirm that GA provides a robust and scalable foundation for simplifying the analysis and control of three-phase systems, paving the way for future theoretical and practical advancements.
\end{abstract}


\section*{Nomenclature}
\begin{IEEEdescription}[\IEEEusemathlabelsep\IEEEsetlabelwidth{SUGAR}]
\item[AC] Alternate current.
\item[GA] Geometric algebra.
\item[MIMO] Multiple-input/multiple-output.
\item[SISO] Single-input/single-output.
\item[R-TF] Real-valued transfer function.
\item[C-TF] Complex-valued transfer function.
\item[GA-TF] Geometric algebra-valued transfer function.
\item[SUGAR] Symbolic and User-friendly Geometric Algebra Routines.
\end{IEEEdescription}

\section{Introduction} 
\label{sec:Intro}


\IEEEPARstart{T}{hree-phase} electrical systems are multivariable systems whose modeling, analysis, and control vary in complexity depending on the chosen mathematical framework~\cite{Har07}. These systems exhibit inherently coupled behavior, where a change in one input may affect multiple outputs, complicating their control~\cite{Sko05}.

In recent years, geometric algebra (GA) has emerged as a promising tool in electrical engineering, offering simplified representations and novel perspectives~\cite{Cha14, Bay21}. Motivated by its potential, this paper introduces GA as a framework for modeling, analyzing, and controlling three-phase systems. The adoption of GA simplifies control analysis and design, enabling reduced-order models and opening new research directions.

This research shows that multivariable three-phase systems, traditionally modeled as real-valued MIMO systems, can be effectively represented as GA-valued SISO systems. Compared to state-of-the-art methods, the GA representation reduces complexity in both system order and linearity. To achieve this, the paper presents a transformation mapping real-valued MIMO models, characterized by standard real-valued transfer functions (R-TFs), to GA-valued SISO models with GA transfer functions (GA-TFs), expressed using multivectors.

The use of GA in control introduces a new perspective for analyzing and designing closed-loop systems. Stability analysis in the GA framework simplifies to studying the roots of a real-valued polynomial, as in traditional SISO systems, allowing direct application of standard stability tools. Additionally, the framework extends classical MIMO techniques, such as Youla-Ku\u{c}era parameterization, to the GA domain, providing a systematic method for designing stabilizing controllers. As an example, the paper shows how this parameterization can be adapted to design controllers that decouple the real-valued MIMO closed-loop system, effectively diagonalizing the transfer matrix.

By utilizing GA, this work provides a unified framework that simplifies three-phase system modeling and control while offering a solid foundation for further advancements in multivariable control.

\subsection{Motivating Reduced-Order Representation}

Fig.~\ref{fig:contributions} sketches existing closed-loop schemes for three-phase electrical systems and the new proposal.
It reveals different complexity issues that arise depending on the underlying mathematical framework.
In all sub-figures, $C_i(\mathrm{p})$ and $G_i(\mathrm{p})$ denote controller and plant transfer functions in different domains.
The figure is fully explained throughout the paper.

As shown in Fig.~\ref{fig:rmimo}, a three-phase dynamic system (and its controller), in the standard $\alpha\beta$ (or $dq$ frame), can be represented by two-phase quantities and modeled as a $2\times 2$ linear real-valued MIMO system, characterized by a matrix of R-TFs (transfer matrix) relating each input to each output.  
An initial model reduction effort is achieved when the two-phase quantities are organized as complex numbers, and these systems can then be represented using complex-valued nonlinear SISO models, as illustrated in Fig.~\ref{fig:cmimo}.  
The plant is described by one or two complex-valued transfer functions (C-TF) depending on whether the system is balanced or unbalanced~\cite{Har20}.
The non-linearity, which appears for unbalanced scenarios where the complex conjugate of the input going through $G_2(\mathrm{p})$ is required, poses difficulties in the analysis and design phases.  

The GA-based approach presented in this work shows that three-phase electrical systems can instead be described by GA-valued linear SISO models, applicable to both balanced and unbalanced scenarios, as shown in Fig.~\ref{fig:gmimo}. In this framework, the plant and controller are characterized by a single GA-TF each, reducing the model complexity with respect to state-of-the-art approaches both in terms of dimensionality and ensuring linearity. This simplification opens opportunities to explore new perspectives for analysis and design within the GA framework, providing a complementary approach to existing methods (Fig.~\ref{fig:rmimo} and~\ref{fig:cmimo}).

The paper illustrates the main contributions through examples, where the symbolic and numeric GA computations have been performed using the Symbolic and User-friendly Geometric Algebra Routines (SUGAR) Matlab toolbox~\cite{Vel24}.
Experimental results corroborate the feasibility of the designs.
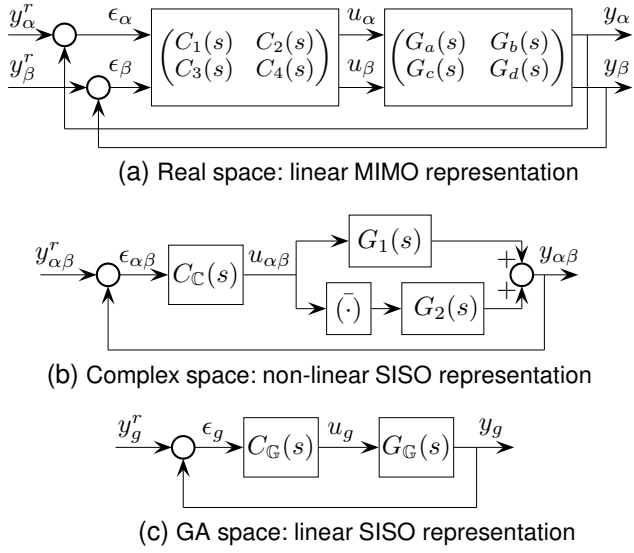
\begin{figure}[!t]
\centering
\subfloat[\small Real space: linear MIMO representation]{
\hspace{-1.35cm}
\begin{pspicture}(11,1.85)
\psframe[linewidth=0.01](3.5,0.55)(6.0,1.85)  \rput(4.8,1.2){\small
$\begin{pmatrix} C_1(\mathrm{p})\! & \!C_2(\mathrm{p}) \\ C_3(\mathrm{p})\! & \!C_4(\mathrm{p}) \end{pmatrix}$
}
\psline[linewidth=0.01,arrowsize=5pt]{->}(1.6,1.5)(2.2,1.5)   \rput(1.8,1.75){$y_{\alpha}^r$}
\pscircle(2.35,1.5){0.15}
\psline[linewidth=0.01,arrowsize=5pt]{->}(2.5,1.5)(3.5,1.5)  \rput(3.1,1.75){$\epsilon_{\alpha}$}
\psline[linewidth=0.01,arrowsize=5pt]{->}(1.6,0.8)(2.65,0.8) \rput(1.8,1.05){$y_{\beta}^r$}
\pscircle(2.8,0.8){0.15}
\psline[linewidth=0.01,arrowsize=5pt]{->}(2.95,0.8)(3.5,0.8)  \rput(3.1,1.05){$\epsilon_{\beta}$}
\psline[linewidth=0.01,arrowsize=5pt]{->}(6,1.5)(6.6,1.5)  \rput(6.3,1.75){$u_{\alpha}$}
\psline[linewidth=0.01,arrowsize=5pt]{->}(6,0.8)(6.6,0.8)  \rput(6.3,1.05){$u_{\beta}$}
\psframe[linewidth=0.01](6.6,0.55)(9.1,1.85)  \rput(7.9,1.2){\small
$\begin{pmatrix} G_a(\mathrm{p})\! & \!G_b(\mathrm{p}) \\ G_c(\mathrm{p})\! & \!G_d(\mathrm{p}) \end{pmatrix}$
}
\psline[linewidth=0.01,arrowsize=5pt]{->}(9.1,1.5)(9.9,1.5)  \rput(9.7,1.75){$y_{\alpha}$}
\psline[linewidth=0.01,arrowsize=5pt]{->}(9.1,0.8)(9.9,0.8)  \rput(9.7,1.05){$y_{\beta}$}
\psline[linewidth=0.01]{-}(9.3,1.5)(9.3,0.25) 
\psline[linewidth=0.01,arrowsize=5pt]{-}(9.3,0.25)(2.35,0.25)  
\psline[linewidth=0.01,arrowsize=5pt]{->}(2.35,0.25)(2.35,1.35) 
\psline[linewidth=0.01]{-}(9.55,0.8)(9.55,0.0) 
\psline[linewidth=0.01,arrowsize=5pt]{-}(9.55,0)(2.8,0)  
\psline[linewidth=0.01,arrowsize=5pt]{->}(2.8,0)(2.8,0.65) 
\end{pspicture}
\label{fig:rmimo}
}
\\
\subfloat[\small Complex space: non-linear SISO representation]{
\hspace{-1cm}\begin{pspicture}(10,1.8)
\psframe[linewidth=0.01](3.5,0.6)(4.5,1.4)  \rput(4.0,1.0){$C_{\mathbb{C}}(\mathrm{p})$}
\psline[linewidth=0.01,arrowsize=5pt]{->}(1.8,1.0)(2.55,1.0) \rput(2,1.25){$y^r_{\alpha\beta}$}
\pscircle(2.7,1.0){0.15}
\psline[linewidth=0.01,arrowsize=5pt]{->}(2.85,1.0)(3.5,1.0)  \rput(3.1,1.2){$\epsilon_{\alpha\beta}$}
\psline[linewidth=0.01,arrowsize=5pt]{-}(4.5,1.0)(5.2,1.0)  \rput(4.85,1.2){$u_{\alpha\beta}$}
\psline[linewidth=0.01,arrowsize=5pt]{-}(5.2,1.0)(5.2,0.55)
\psline[linewidth=0.01,arrowsize=5pt]{->}(5.2,0.55)(5.6,0.55)
\psframe[linewidth=0.01](5.6,0.2)(6.2,0.9)  \rput(5.9,0.55){$\bar{(\cdot)}$}
\psline[linewidth=0.01,arrowsize=5pt]{->}(6.2,0.55)(6.6,0.55) 
\psframe[linewidth=0.01](6.6,0.2)(7.7,0.9)  \rput(7.2,0.55){$G_{2}(\mathrm{p})$}
\psline[linewidth=0.01,arrowsize=5pt]{-}(7.7,0.55)(8.2,0.55) 
\psline[linewidth=0.01,arrowsize=5pt]{->}(8.2,0.55)(8.2,0.85) 
\pscircle(8.2,1.0){0.15}
 \rput(8.0,1.2){$+$}
  \rput(8.0,0.8){$+$}

\psline[linewidth=0.01,arrowsize=5pt]{-}(5.2,1.0)(5.2,1.45)
\psline[linewidth=0.01,arrowsize=5pt]{->}(5.2,1.45)(5.9,1.45)

\psframe[linewidth=0.01](5.9,1.1)(7.0,1.8)  \rput(6.45,1.45){$G_{1}(\mathrm{p})$}
\psline[linewidth=0.01,arrowsize=5pt]{-}(7.0,1.45)(8.2,1.45)
\psline[linewidth=0.01,arrowsize=5pt]{->}(8.2,1.45)(8.2,1.15)

\psline[linewidth=0.01,arrowsize=5pt]{->}(8.35,1.0)(8.95,1.0)  \rput(8.7,1.25){$y_{\alpha\beta}$}
\psline[linewidth=0.01]{-}(8.5,1.0)(8.5,0.0) 
\psline[linewidth=0.01,arrowsize=5pt]{-}(8.5,0.0)(2.7,0.0)  
\psline[linewidth=0.01,arrowsize=5pt]{->}(2.7,0.0)(2.7,0.85) 
\end{pspicture}
\label{fig:cmimo}
}
\\
\subfloat[\small GA space: linear SISO representation]{
\begin{pspicture}(10,1.2)
\psframe[linewidth=0.01](3.5,0.4)(4.5,1.2)  \rput(4,0.8){$C_{\mathbb{G}}(\mathrm{p})$}
\psline[linewidth=0.01,arrowsize=5pt]{->}(1.8,0.8)(2.55,0.8) \rput(2,1.05){$y^r_g$}
\pscircle(2.7,0.8){0.15}
\psline[linewidth=0.01,arrowsize=5pt]{->}(2.85,0.8)(3.5,0.8)  \rput(3.1,1.0){$\epsilon_g$}
\psline[linewidth=0.01,arrowsize=5pt]{->}(4.5,0.8)(5.3,0.8)  \rput(4.8,1.0){$u_g$}
\psframe[linewidth=0.01](5.3,0.4)(6.3,1.2)  \rput(5.8,0.8){$G_{\mathbb{G}}(\mathrm{p})$}
\psline[linewidth=0.01,arrowsize=5pt]{->}(6.3,0.8)(7.1,0.8)  \rput(6.8,1.05){$y_g$}
\psline[linewidth=0.01]{-}(6.6,0.8)(6.6,0.0) 
\psline[linewidth=0.01,arrowsize=5pt]{-}(6.6,0.0)(2.7,0.0)  
\psline[linewidth=0.01,arrowsize=5pt]{->}(2.7,0.0)(2.7,0.65) 
\end{pspicture}
\label{fig:gmimo}
}
\caption{Control of three-phase electrical systems in different spaces}
\label{fig:contributions}
\vspace{-0.35cm}
\end{figure}

\subsection{Related Work}
The application of GA to the electrical engineering field is not new and in AC circuit analysis is even common~\cite{Hit24}.
But it has been mainly bounded to re-define spare concepts such as the (complex) frequency~\cite{Mil22},  in the analysis of the apparent power~ \cite{Men07, Cas08, Cas10, Cas12, Mon21a, Mon21}, or in the analysis of power flow~\cite{Mon22, Eid22, Mon23}.
Moreover, in the context of three-phase electrical systems, some of the previous results such as~\cite{Mon21a} focus on the signal analysis side  while the current paper targets the system dynamics side.
This difference implies that previous works cover the steady state analysis while the current paper offers an extension also to the transient dynamics. 

In the field of system dynamics and control (apart from the robotics discipline where GA has been widely applied, see eg.~\cite{Low23}), the application of GA is starting to occur.   Second order systems expressed in terms of generalized coordinates is investigated using  GA language in~\cite{Gar01}. Well-known Lyapunov stability conditions and sliding mode control conditions are re-visited in terms of GA in \cite{Sir22} and \cite{Sir23}, respectively. 

Previous works differs from the current paper domain which is  closed-loop control of three-phase electrical systems. The modeling effort presented in this paper is not bounded only to second order systems, and the presented control tools in the newly defined GA framework belong to the discipline of control of linear SISO systems.

\subsection{Summary of Contributions and Paper Structure}

The main contributions of this work regarding the modeling and analysis of three-phase electrical systems, for both balanced and unbalanced scenarios, can be summarized as follows:
\begin{itemize}
\item Definition of the GA framework, establishing a foundation for applying GA to three-phase systems.
\item New GA-valued linear SISO model, reducing the complexity of multivariable systems.
\item As an example of how classical MIMO theories can be seamlessly extended to GA, this adapts the construction of all stabilizing controllers to the GA domain.
\item Using this adaptation, the complete family of GA decoupling controllers is derived, illustrating the potential of the framework to address fundamental MIMO problems in a simpler manner.
\end{itemize}

This work focuses on theoretical analysis, demonstrating how expanding the design space from real to GA-valued domains simplifies tools and makes problems more mathematically tractable (e.g., reducing nonlinearities).
Working in a higher-dimensional space inherently requires new theoretical approaches, such as GA, which may initially appear challenging.
However, these efforts are rewarded with significant benefits in terms of simplicity and flexibility.

The GA techniques introduced in this study offer new avenues for improving the analysis and control of three-phase electrical systems. They enable the development of alternative tools that complement existing methods, simplifying and enhancing analysis and design, particularly in scenarios where traditional frameworks lack intuitiveness.

The paper is structured as follows. Section \ref{ss:back} reviews real and complex-valued models of three-phase electrical systems. Section \ref{ss:geometric} presents the GA framework and the new GA model. Sections \ref{ss:stability} and \ref{ss:youla} discuss closed-loop stability and controller design in the GA framework. Section \ref{ss:res} presents laboratory experiments, and Section \ref{ss:con} concludes the paper.

{\bf Notation.}
$\mathbb{C}^n$ and $\mathbb{R}^n$ denote the complex and real $n$-dimensional space; $\bar x$ denotes the conjugate of a complex vector $x\in\mathbb{C}^n$; $j\in \mathbb{C}$ is the imaginary number such that $j^2=-1$; $\mathbb{F}$ denotes the R-TF space and $G(\mathrm{p})\in \mathbb{F}$ denotes a R-TF,  where the standard TF Laplace argument $\mathrm{s}\in\mathbb{C}$ has been replaced by $\mathrm{p}$ (further defined below). For notation convenience, the space of R-TFs will also be denoted by the GA description given by $\mathcal{F}_{0,0}$, which is explained in detail in the Appendix.
Then, $G_{\mathbb{R}}(\mathrm{p})=G(\mathrm{p})\in \mathcal{F}_{0,0}$ denotes a R-TF, i.e., $\mathbb{F}$ and $\mathcal{F}_{0,0}$ are interchangeable.
Similarly, the space of C-TFs is denoted by $\mathcal{F}_{0,1}$ and $G_{\mathbb{C}}(\mathrm{p})=G_a(\mathrm{p})+jG_b(\mathrm{p})\in \mathcal{F}_{0,1}$ denotes a C-TF, with $G_a(\mathrm{p}), G_b(\mathrm{p})\in{\mathcal{F}_{0,0}}$.

Finally, the space of GA-TF is denoted by $\mathcal{F}_{2,0}$ where $G_{\mathbb{G}}(\mathrm{p})=G_a(\mathrm{p})e_0+G_b(\mathrm{p})e_1+G_c(\mathrm{p})e_2+G_d(\mathrm{p})e_{12}\in \mathcal{F}_{2,0}$ denotes a GA-TF, with $G_a(\mathrm{p}), G_b(\mathrm{p}), G_c(\mathrm{p}), G_d(\mathrm{p}) \in \mathcal{F}_{0,0}$ and  $e_0$, $e_1$, $e_2$, and $e_{12}$ denote  its basis elements.
In all transfer functions, $\mathrm{p}\in \mathcal{F}_{2,0}$. The GA conjugate of $G_{\mathbb{G}}(\mathrm{p})\in \mathcal{F}_{2,0}$ is given by $\bar G_{\mathbb{G}}(\mathrm{p})=G_a(\mathrm{p})e_0-G_b(\mathrm{p})e_1-G_c(\mathrm{p})e_2-G_d(\mathrm{p})e_{12}$. The term $\underline x$  denotes the dual of the GA element $x\in\mathcal{F}_{p,q}$.
The meaning of the  $p, q$ values accompanying each space $\mathcal{F}_{p,q}$ is given in the Appendix. Transfer function matrices whose entries are R-TFs, C-TFs or GA-TFs belong to particular $\mathcal{F}_{p,q}^{n \times m}$ spaces. 

\section{Real-valued and Complex-valued Models}\label{ss:back}
State-of-the-art real-valued and complex-valued approaches to the modeling of three-phase electrical systems are reviewed.
Among the numerous systems that could serve as examples, including those with capacitors and inductors, we selected a model that can be easily transformed between balanced and unbalanced configurations.
To focus on the core aspects of the comparison, we deliberately chose a simple model that effectively illustrates the differences.

\begin{example}[Illustrative three-phase electrical system]\label{e:ex1}

 \begin{figure}
 	\centering
	\includegraphics[width=5cm]{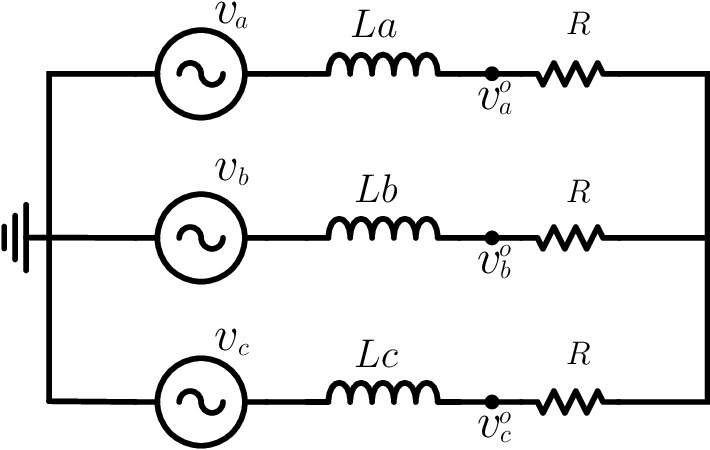}	
	\caption{Scheme for the three-phase electrical system used in the examples.}
	\label{fig:scheme3}
\end{figure}

Fig.~\ref{fig:scheme3} shows the system used throughout the paper and provides the experimental setup to illustrate how to deploy a GA controller in a real scenario.  It is composed by three ideal voltage sources that feed a load~$R$, over a transmission line with  inductances $L_a$, $L_b$ and $L_c$. Depending on the values of the inductances, the system is balanced (when $L_a=L_b=L_c=L$) or unbalanced (for example when $L_a=L_c=L$ and $L_b=L_u$). Whenever numerical values are required, they correspond to the configuration given by $L=3\cdot 10^{-3}$H, $L_u=3\cdot 10^{-2}$H, $R=22\Omega$, where voltage sinusoidal signals have an amplitude of $V=155$ V and a frequency of $\omega=2\pi60$ rad/s.

The system's inputs (or the control actions when a controller is connected) are the three voltages \(v_a(t)\), \(v_b(t)\), and \(v_c(t)\). The outputs are the load voltages \(v_a^o(t)\), \(v_b^o(t)\), and \(v_c^o(t)\). The system will be analyzed in both open-loop and closed-loop configurations under various scenarios. In the closed-loop evaluation, the controller's objective will be explicitly defined.

The example presented in this work is intentionally simple to illustrate how the fundamental principles of the GA framework apply to modelling and analysing three-phase electrical systems. Its adaptability to both balanced and unbalanced cases ensures broad applicability, while its simplicity allows for easy reproduction in laboratory settings. This approach emphasizes key comparative aspects and computational implications without the complexity of more intricate examples.
\end{example}
\subsection{Real-valued Representation}\label{ss:rvrep}
Three-phase quantities can be modeled as equivalent two-phase quantities (if zero-sequence components are disregarded). For example, in the $\alpha\beta$ frame (the same applies to the $dq$), three-phase electrical systems can be represented by 
\begin{equation}\label{eq:rmimo}
\left(\begin{array}{c} y_{\alpha}(\mathrm{p}) \\ y_{\beta}(\mathrm{p}) \end{array}\right)=\underbrace{\left(\begin{array}{cc} G_a(\mathrm{p}) & G_b(\mathrm{p}) \\ G_c(\mathrm{p}) & G_d(\mathrm{p}) \end{array}\right)}_{M_{\mathbb{R}}(\mathrm{p})}\left(\begin{array}{c} u_{\alpha}(\mathrm{p}) \\ u_{\beta}(\mathrm{p}) \end{array}\right)
\end{equation}
where the system matrix $M_{\mathbb{R}}(\mathrm{p})$ elements are R-TFs denoted by $G_a(\mathrm{p}),G_b(\mathrm{p}), G_c(\mathrm{p}),G_d(\mathrm{p})\in  \mathcal{F}_{0,0}$, and $u=\left(\begin{array}{cc}u_{\alpha}(\mathrm{p}) &  u_{\beta}(\mathrm{p})\end{array}\right)^T\in \mathcal{F}_{0,0}^{2 \times 1}$ and $y=\left(\begin{array}{cc}y_{\alpha}(\mathrm{p}) &  y_{\beta}(\mathrm{p})\end{array}\right)^T\in \mathcal{F}_{0,0}^{2 \times 1}$ are the input  and output vectors, respectively.  Eq.~(\ref{eq:rmimo}) will be referred to as real-valued MIMO model and corresponds to the plant in the closed-loop scheme of Fig.~\ref{fig:rmimo}.

For three-phase balanced systems, $M_{\mathbb{R}}$ in~(\ref{eq:rmimo}) has a specific structure with $G_d(\mathrm{p})=G_a(\mathrm{p})$ and $G_b(\mathrm{p})=G_c(\mathrm{p})=0$, that is,  it  is diagonal, which implies that the $\alpha\beta$ channels are decoupled.  Hence, a MIMO diagonal controller can be designed considering two independent loops and using linear SISO tools and the real-valued controller matrix in the closed-loop scheme of Fig.~\ref{fig:rmimo} will have $C_2(\mathrm{p})=C_3(\mathrm{p})=0$. For the unbalanced case, the controller design problem becomes more complex because system matrix $M_{\mathbb{R}}(\mathrm{p})$ in (\ref{eq:rmimo}) is fully populated,  thus requiring to apply MIMO design techniques.

\begin{example}[Real-valued MIMO model]\label{e:ex2}
Applying standard modeling techniques such as the  modified nodal analysis~\cite{Wed02}, the system matrix $M_{\mathbb{R}}(\mathrm{p})$ of the real-valued MIMO model (\ref{eq:rmimo}) of the system of Fig. \ref{fig:scheme3} can be obtained. For the balanced case,   the system matrix is given by
\begin{equation}\label{eq:rmimoexbalanced}\small
M_{\mathbb{R}}(\mathrm{p})=\left(\begin{array}{cc}  R(L\mathrm{p}+R)^{-1}& 0 \\ 0 & R(L\mathrm{p}+R)^{-1}  \end{array}\right)
\end{equation}
and for the unbalanced case by 
\begin{equation}\label{eq:rmimoexunbalanced}\small
M_{\mathbb{R}}(\mathrm{p})\!\!=\!\!\left(\!\!\begin{array}{cc} 3R(2R\!+\!L\mathrm{p}\!+\!L_u\mathrm{p})d(\mathrm{p})^{-1} \!\!&\!\! -\sqrt{3}R\mathrm{p}(L\!-\!L_u)d(\mathrm{p})^{-1} \\  -\sqrt{3}R\mathrm{p}(L\!-\!L_u)d(\mathrm{p})^{-1}\!\!&\!\! R(6R\!+\!5L\mathrm{p}\!+\!L_u\mathrm{p})d(\mathrm{p})^{-1} \end{array}\!\!\!\!\right)
\end{equation}
where 
\begin{equation}\label{eq:rmimoexunbalancedden}
d(\mathrm{p})=2(R+L\mathrm{p})(3R+L\mathrm{p}+2L_u\mathrm{p})
\end{equation}

\end{example}

\subsection{Complex-valued Representation}

The use of complex-valued dynamical models to represent  three-phase electrical systems has been widely used. The linear transformation $T_\mathbb{C}:\mathcal{F}_{0,0}^{2\times 1}\rightarrow \mathcal{F}_{0,1}^{2\times 1}$ defined by
\begin{equation}\label{eq:transrc}
\left(\!\begin{array}{c} x_{\alpha\beta}(\mathrm{p}) \\ \bar x_{\alpha\beta}(\mathrm{p}) \end{array}\!\right)=T_{\mathbb{C}}\left(\!\begin{array}{c} x_{\alpha}(\mathrm{p}) \\ x_{\beta}(\mathrm{p}) \end{array}\!\right)\; \text{with }\;
T_{\mathbb{C}}=\left(\begin{array}{cc} 1 & j \\ 1 & -j \end{array}\right)
\end{equation}
where   $x_{\alpha}(\mathrm{p}), x_{\beta}(\mathrm{p}) \in \mathcal{F}_{0,0}$ are R-TFs and $x_{\alpha\beta}(\mathrm{p})=x_{\alpha}(\mathrm{p})+jx_{\beta}(\mathrm{p}) \in \mathcal{F}_{0,1}$ is a C-TF, allows transferring real-valued MIMO systems~(\ref{eq:rmimo}) to the complex domain \cite{Und76}, leading to
\begin{equation}\label{eq:cmimo}
\left(\begin{array}{c} y_{\alpha\beta}(\mathrm{p}) \\ \bar y_{\alpha\beta}(\mathrm{p}) \end{array}\right)=\underbrace{\left(\begin{array}{cc} G_1(\mathrm{p}) & G_2(\mathrm{p}) \\ \bar G_2(\mathrm{p}) & \bar G_1(\mathrm{p}) \end{array}\right)}_{M_{\mathbb{C}}(\mathrm{p})}\left(\begin{array}{c} u_{\alpha\beta}(\mathrm{p})  \\ \bar u_{\alpha\beta}(\mathrm{p})  \end{array}\right)
\end{equation}
where the elements of the system matrix $M_{\mathbb{C}}(\mathrm{p})$ are C-TFs denoted by $G_1(\mathrm{p}), G_2(\mathrm{p})\in \mathcal{F}_{0,1}$ and given by
\begin{equation}
\begin{aligned}\label{eq:g1g2}
G_1(\mathrm{p})=&\frac{G_a(\mathrm{p})+G_d(\mathrm{p})}{2}+j\frac{-G_b(\mathrm{p})+G_c(\mathrm{p})}{2}  \\
G_2(\mathrm{p})=&\frac{G_a(\mathrm{p})-G_d(\mathrm{p})}{2}+j\frac{G_b(\mathrm{p})+G_c(\mathrm{p})}{2} 
\end{aligned}
\end{equation}
with   $u_{\alpha\beta}(\mathrm{p}), y_{\alpha\beta}(\mathrm{p})\in \mathcal{F}_{0,1}$. Henceforth, this representation will be referred to as complex-valued MIMO model. 

The obtained MIMO model (\ref{eq:cmimo})-(\ref{eq:g1g2}) exhibits a symmetry:  the  dynamics of the second output correspond to the first one in conjugate form.  As a consequence, its  analysis  simplifies to consider only the complex-valued SISO system given by
\begin{equation}\label{eq:asymsiso}
y_{\alpha\beta}(\mathrm{p})=G_1(\mathrm{p})u_{\alpha\beta}(\mathrm{p})+G_2(\mathrm{p})\bar u_{\alpha\beta}(\mathrm{p})
\end{equation}
which corresponds to the  closed-loop scheme of Fig.~\ref{fig:cmimo}, which is non-linear due to the fact that $G_2(\mathrm{p})$ is multiplied by the conjugate of the input, $\bar u_{\alpha\beta}(\mathrm{p})$. The order of the system, from (\ref{eq:rmimo}) to (\ref{eq:asymsiso}), has been reduced by half at the expenses of introducing the nonlinearity that challenges the design of the C-TF controller, $C_{\mathbb{C}}(\mathrm{p})$. Note that for balanced systems it holds that $G_2(\mathrm{p})=0$  and the complex-value SISO system (\ref{eq:asymsiso}) becomes linear, thus   allowing using extensions of  linear SISO techniques to   complex-valued models (e.g. \cite{Dor16, Bae19, Dor21}).

\begin{example}[Complex-valued MIMO and SISO]\label{e:ex3}
Under transformation (\ref{eq:transrc}), the real  MIMO model  given in (\ref{eq:rmimoexbalanced})-(\ref{eq:rmimoexunbalanced}) transforms to~(\ref{eq:cmimo})-(\ref{eq:g1g2}) where the system matrix $M_{\mathbb{C}}(\mathrm{p})$ is given by  
\begin{equation}\label{eq:cmimoexbalanced}\small
M_{\mathbb{C}}(\mathrm{p})=\left(\begin{array}{cc} R(L\mathrm{p}+R)^{-1}  & 0 \\ 0 & R(L\mathrm{p}+R)^{-1} \end{array}\right)
\end{equation}
for the balanced case, and for the unbalanced case by 
\begin{equation}\label{eq:cmimoexunbalanced}\small
M_{\mathbb{C}}(\mathrm{p})\!\!=\!\!\left(\!\!\begin{array}{cc} \!\! 2R(3R\!+\!2L\mathrm{p}\!+\!L_u\mathrm{p})d(\mathrm{p})^{-\!1}\! \! \! \! \!&\! \! \! \! \! g_{r}(\mathrm{p})\!+\!g_{i}(\mathrm{p})j\! \\\! g_{r}(\mathrm{p})\!+\!g_{i}(\mathrm{p})j\!  \!&\! \! 2R(3R\!+\!2L\mathrm{p}\!+\!L_u\mathrm{p})d(\mathrm{p})^{-\!1}\!\!\end{array}\!\!\!\!\right)
\end{equation}
where
\begin{equation} \small
g_{r}(\mathrm{p})\!=\!R\mathrm{p}(L\!-\!L_u)d(\mathrm{p})^{-1} \text{ and }
g_{i}(\mathrm{p})\!=\!-\sqrt{3}R\mathrm{p}(L\!-\!L_u)d(\mathrm{p})^{-1}
\end{equation}

with $d(\mathrm{p})$ given in (\ref{eq:rmimoexunbalancedden}). And thanks to  the symmetry, and according to (\ref{eq:asymsiso}), the control problem  for the balanced case only considers the linear complex-valued SISO system characterized by~(\ref{eq:cmimoexbalanced}) while for the unbalanced case it must consider the nonlinear complex-valued SISO system characterized by~(\ref{eq:cmimoexunbalanced}).
\end{example}

\section{Geometric Algebra Representation}\label{ss:geometric}

Three-phase electrical systems will be represented in a new mathematical domain characterized by a 4-dimensional geometric algebra (GA) denoted as $\mathcal{F}_{2,0}$ (further explained in the Appendix). This algebra is spanned by four basis elements: $e_0$, $e_1$, $e_2$, and $e_{12}$. Each element of $\mathcal{F}_{2,0}$ is a multivector that can be expressed as a linear combination of these basis elements, $e_0$, $e_1$, $e_2$, and $e_{12}$, with coefficients that are real-time functions (R-TFs). Specifically, if a multivector $V$ belongs to $\mathcal{F}_{2,0}$, it can be expressed as $V = a_0e_0 + a_1e_1 + a_2e_2 + a_{12}e_{12}$, where $a_0e_0$ is called the real part of $V$, $a_1e_1 + a_2e_2 + a_{12}e_{12}$ is called the vector part of $V$, and $e_{12}$ is called the pseudoscalar, which satisfies the property that it squares to $-1$.

The representation of the real-valued MIMO system~(\ref{eq:rmimo}) in the new GA domain is also obtained by applying a particular transformation. The transformation has been chosen to achieve a decoupled and symmetric representation both for balanced and unbalanced systems, which  brings the analysis to the GA-value linear SISO systems domain. 
Specifically, the transformation, namely $T_{\mathbb{G}}$,  should be able to transform the $2 \times 2$  system matrix $M_{\mathbb{R}}(\mathrm{p})$ characterizing  the real-valued MIMO model (\ref{eq:rmimo}) into an equivalent $2 \times 2$ system matrix having a diagonal form, and with equal diagonal elements (equal GA-TFs). The latter will ensure that only one of the two diagonal elements will be used for systems' analysis, which reduces the order to a $1 \times 1$  system matrix with a single GA-TF. Hence,  the desired transformation $T_{\mathbb{G}}:\mathcal{F}_{0,0}^2\rightarrow \left(\mathcal{F}_{2,0}\right)^2$ should fulfill
\begin{equation}\label{eq:eqtrans}
T_{\mathbb{G}}M_{\mathbb{R}}(\mathrm{p})T_{\mathbb{G}}^{-1}=M_{\mathbb{G}}(\mathrm{p})
\end{equation}
where $M_{\mathbb{G}}(\mathrm{p})$ will be the new system matrix characterizing the novel GA-valued model. And since $M_{\mathbb{G}}(\mathrm{p})$ is set to be diagonal with equal diagonal elements, solving (\ref{eq:eqtrans}) and finding $T_{\mathbb{G}}$ will ultimately imply simplifying the whole analysis, specifically for the unbalanced case, avoiding either linear MIMO~(\ref{eq:rmimo}) or non-linear SISO~(\ref{eq:asymsiso}) models. 

The original real-valued MIMO model  (\ref{eq:rmimo}) is transferred to the new  $\mathcal{F}_{2,0}$ GA space by applying the  linear transformation  $T_{\mathbb{G}}$
defined by
\begin{equation}\label{eq:transrg}
\left(\!\!\begin{array}{c} x_{g}(\mathrm{p}) \\ \underline x_{g}(\mathrm{p}) \end{array}\!\!\right)\!=\!T_{\mathbb{G}}\!\left(\!\!\begin{array}{c} x_{\alpha}(\mathrm{p}) \\ x_{\beta}(\mathrm{p}) \end{array}\!\!\right) \text{with } 
T_{\mathbb{G}}\!=\!\frac{1}{2}\!\left(\!\!\!\begin{array}{cc} e_0\!+\!e_1\!\!&\!\! -e_2\!+\!e_{12} \\ -e_2\!-\!e_{12} \!\!&\!\! e_0\!-\!e_1\end{array}\!\!\!\right)
\end{equation}
that solves (\ref{eq:eqtrans}), where 
$x_{\alpha}(\mathrm{p}), x_{\beta}(\mathrm{p})\in \mathcal{F}_{0,0}$ are R-TFs, $x_{g}(\mathrm{p})=\frac{1}{2}(x_{\alpha}(\mathrm{p})e_0+x_{\alpha}(\mathrm{p})e_1-x_{\beta}(\mathrm{p})e_2+x_{\beta}(\mathrm{p})e_{12}\in \mathcal{F}_{2,0}$ is a GA-TF, and $\underline x_{g}(\mathrm{p})\in \mathcal{F}_{2,0}$ stands for the dual of $x_{g}(\mathrm{p})$ and it is given by $\underline x_{g}(\mathrm{p})=x_{g}(\mathrm{p})e_{12}$~\cite{Hes84}. By applying~(\ref{eq:transrg}) to  (\ref{eq:rmimo}) leads to
\begin{equation}\label{eq:gmimo}
\left(\begin{array}{c} y_{g}(\mathrm{p}) \\ \underline y_{g}(\mathrm{p}) \end{array}\right)=\underbrace{\left(\begin{array}{cc} G_{\mathbb{G}}(\mathrm{p}) & 0 \\ 0 & G_{\mathbb{G}}(\mathrm{p}) \end{array}\right)}_{M_{\mathbb{G}}(\mathrm{p})}\left(\begin{array}{c} u_{g}(\mathrm{p}) \\ \underline u_{g}(\mathrm{p}) \end{array}\right)
\end{equation}
where $u_{g}(\mathrm{p}), y_{g}(\mathrm{p})\in \mathcal{F}_{2,0}$, and  the diagonal elements  $G_{\mathbb{G}}(\mathrm{p})\in \mathcal{F}_{2,0}$ of the system matrix $M_{\mathbb{G}}(\mathrm{p})$ are the GA-TFs 
\begin{equation}\label{eq:gs3}
\begin{aligned}
G_{\mathbb{G}}(\mathrm{p})=&\frac{1}{2}\left(G_a(\mathrm{p})+G_d(\mathrm{p})e_0 +G_a(\mathrm{p})-G_d(\mathrm{p})e_1\right.\\
&+\left.G_b(\mathrm{p})+G_c(\mathrm{p})e_2+G_b(\mathrm{p})-G_c(\mathrm{p})e_{12} \right)
\end{aligned}
\end{equation}
and the zeros of the contra-diagonal stand for $0=0e_0+0e_1+0e_2+0e_{12}\in \mathcal{F}_{2,0}$. Henceforth, this representation will be referred to as GA-valued MIMO model.  

The GA MIMO model~(\ref{eq:gmimo}) exhibits a decoupled structure  ($M_{\mathbb{G}}(\mathrm{p})$ is diagonal), and since both diagonal elements are equal, its  analysis simplifies to consider only the  GA-valued linear SISO system 
\begin{equation}\label{eq:gsiso}
y_{g}(\mathrm{p})=G_{\mathbb{G}}(\mathrm{p})u_{g}(\mathrm{p})
\end{equation}
which is the plant of  Fig.~\ref{fig:gmimo}. From  $M_{\mathbb{G}}(\mathrm{p})$ in~(\ref{eq:gmimo}), it holds 
\begin{equation}\label{eq:selS}
G_{\mathbb{G}}(\mathrm{p})=S^TM_{\mathbb{G}}(\mathrm{p})S, \text{ with } S=\left(\begin{array}{c} 1 \\ 0 \end{array}\right)
\end{equation}

\begin{example}[GA-valued MIMO and SISO]\label{e:ex4}
Under transformation (\ref{eq:transrg}), the real-valued MIMO model   (\ref{eq:rmimoexbalanced})-(\ref{eq:rmimoexunbalanced}) transforms to~(\ref{eq:gmimo})  where the system matrix $M_{\mathbb{G}}(\mathrm{p})$ is defined by 
\begin{equation}\label{eq:gmimoexbalanced}\small
M_{\mathbb{G}}(\mathrm{p})=\left(\begin{array}{cc} R(L\mathrm{p}+R)^{-1}e_0 & 0 \\ 0 & R(L\mathrm{p}+R)^{-1}e_0 \end{array}\right)
\end{equation}
for the balanced case, and for the unbalanced case by 
\begin{equation}\label{eq:gmimoexunbalanced}\small
M_{\mathbb{G}}(\mathrm{p})\!\!=\!\!\left(\!\!\begin{array}{cc} \!g_0(\mathrm{p})e_0\!+\!g_1(\mathrm{p})e_1\!\!+\!g_2(\mathrm{p})e_2 \! \!\!\!&\!\!\! \!0 \! \! \\ \! \! 0 \!\!\!\!&\!\!\!\! g_0(\mathrm{p})e_0\!+\!g_1(\mathrm{p})e_1\!\!+\!g_2(\mathrm{p})e_2 \! \!\end{array}\!\!\!\right)
\end{equation}
where
\begin{equation}
\begin{aligned}
g_0(\mathrm{p})&=2R(3R+2L\mathrm{p}+L_u\mathrm{p})d(\mathrm{p})^{-1}\\
g_1(\mathrm{p})&=-R(L-L_u)\mathrm{p}d(\mathrm{p})^{-1} \\ 
g_2(\mathrm{p})&=\sqrt{3}R(L_u-L)\mathrm{p}d(\mathrm{p})^{-1}
\end{aligned}
\end{equation}

with $d(\mathrm{p})$ given in (\ref{eq:rmimoexunbalancedden}). Hence, the balanced case must consider a GA-valued SISO system (\ref{eq:gsiso}) characterized by (\ref{eq:gmimoexbalanced}) while the unbalanced case must also consider a GA-valued SISO system~(\ref{eq:gsiso}) but characterized by (\ref{eq:gmimoexunbalanced}).

\end{example}

\section{Discussion}\label{ss:discursion}

Now that the different elements of a control loop have been established, we can analyze the key differences between the RV, CV, and GAV models for balanced and unbalanced systems. This comparison focuses on their computational characteristics, input-output structures, and the nature of the operations required to compute closed-loop expressions. Thus, providing a concise overview of the advantages and limitations of each approach.

In balanced systems, the complex-valued (CV) model offers clear advantages due to its scalar nature for both inputs and outputs, and the commutativity of the operations required to compute closed-loop expressions, which simplifies the overall process. However, as systems become unbalanced, the CV model becomes impractical due to the nonlinearities and the loss of commutativity in these operations. Conversely, the geometric algebra (GA) model maintains scalar inputs and outputs, significantly simplifying controller design. While the operations required to compute closed-loop expressions in the GA model are inherently non-commutative, no existing framework achieves commutativity under these conditions, making GA a robust and viable alternative.

The proposed GA framework facilitates straightforward computation of controllers. However, its novelty means that standard tools for performing common tasks in controller design are not yet widely available. The next section addresses this limitation by adapting basic design tools from real-valued (RV) systems to the GA framework. It includes an exploration of stability analysis in GA-based systems and demonstrates, through an example, how traditional RV design methods can be translated to GA. Furthermore, it considers the practical challenges of implementing these controllers in real systems, showcasing a design example using these adapted techniques.

\section{Migrating Analisis and design tools from RV into GA}\label{ss:analisisanddesign}
The essential work of an engineer relies on a set of tools that are not yet fully available in the GA framework. Among these, stability analysis (both for open-loop and closed-loop systems) is indispensable and is addressed in the first subsection. On the other hand, design tools play a critical role in enabling the synthesis of controllers that meet specific performance criteria. While the vast array of existing tools requires a gradual migration to this framework, this section provides an example by adapting the Youla--Kučera \cite{You76} parametrization to illustrate the simplicity and potential of controller design in GA. This serves as a starting point, as other design techniques can also be migrated with relative ease.

\subsection{GA Stability Analysis}\label{ss:stability}
The system analysis and controller design problem is reduced to a linear GA-valued SISO model~(\ref{eq:gsiso}) 
characterized by the   plant and controller GA-TFs, $G_{\mathbb{G}}(\mathrm{p}), C_{\mathbb{G}}(\mathrm{p})\in \mathcal{F}_{2,0}$,  whose generic expressions can be written as
\begin{equation}\label{eq:gplantandcontroller}
\begin{aligned}
G_{\mathbb{G}}(\mathrm{p})&=g_0(\mathrm{p})e_0+g_1(\mathrm{p})e_1+g_2(\mathrm{p})e_2+g_3(\mathrm{p})e_{12}\\
C_{\mathbb{G}}(\mathrm{p})&=c_0(\mathrm{p})e_0+c_1(\mathrm{p})e_1+c_2(\mathrm{p})e_2+c_3(\mathrm{p})e_{12}
\end{aligned}
\end{equation}

To analyze closed-loop stability,  the closed-loop GA-TF  corresponding to Fig.~\ref{fig:gmimo} is 
\begin{equation}\label{eq:closedloop}
G_{\mathbb{G}}^{cl}(\mathrm{p})=G_{\mathbb{G}}(\mathrm{p})C_{ \mathbb{G}}(\mathrm{p})\left(e_0+G_{\mathbb{G}}(\mathrm{p})C_{ \mathbb{G}}(\mathrm{p})\right)^{-1}
\end{equation}
Assume that the plant and controller (\ref{eq:gplantandcontroller}) are expressed in a numerator/denominator structure as follows
\begin{equation}\label{eq:gplantcontroller}
G_{\mathbb{G}}(\mathrm{p})=\frac{n_p(\mathrm{p})}{d_p(\mathrm{p})} \quad \text{and} \quad C_{\mathbb{G}}(\mathrm{p})=\frac{n_c(\mathrm{p})}{d_c(\mathrm{p})}
\end{equation}
where $n_{p}(\mathrm{p}),d_{p}(\mathrm{p}),n_{c}(\mathrm{p}),d_{c}(\mathrm{p})\in \mathcal{F}_{2,0}$, and specifically, where the denominators $d_p(\mathrm{p})$ and $d_c(\mathrm{p})$ have only scalar part. Note that it is always possible to manipulate the quotients given in (\ref{eq:gplantcontroller}) is such a way  that they can be expressed in terms of  denominators having only scalar part, thus defined by real-valued polynomial functions, $d_p(\mathrm{p}),d_c(\mathrm{p})\in\mathbb{F}\equiv\mathcal{F}_{0,0}$.

\begin{proposition}\label{p:stability}
The closed-loop scheme shown in Fig.\ref{fig:gmimo}, whose GA-TF is given in (\ref{eq:closedloop}), and where the controller and  plant  are denoted by $C_{ \mathbb{G}}(\mathrm{p}), G_{ \mathbb{G}}(\mathrm{p})\in  \mathcal{F}_{2,0}$, and decomposed as in (\ref{eq:gplantcontroller}), is asymptotically stable if the  roots of a real-valued polynomial function $d_{cl}(\mathrm{p})\in \mathcal{F}_{0,0}$ given by
\begin{equation}\label{eq:stability}
d_{cl}(\mathrm{p})=\overline{ d_{pc}(\mathrm{p})}d_{pc}(\mathrm{p})
\end{equation}
with $d_{pc}(\mathrm{p})=d_p(\mathrm{p})d_c(\mathrm{p})+n_p(\mathrm{p})n_c(\mathrm{p})\in  \mathcal{F}_{2,0}$, have negative real part.
\end{proposition}

\begin{proof}
By using (\ref{eq:gplantcontroller}) and a few algebraic operations, the closed-loop transfer function (\ref{eq:closedloop}) can be further written as 
\begin{equation}\label{eq:proof2}
\begin{aligned}
G_{\mathbb{G}}^{cl}(\mathrm{p})\!=&\frac{n_p(\mathrm{p})}{d_p(\mathrm{p})}\frac{n_c(\mathrm{p})}{d_c(\mathrm{p})}\left(e_0+\frac{n_p(\mathrm{p})}{d_p(\mathrm{p})}\frac{n_c(\mathrm{p})}{d_c(\mathrm{p})}\right)^{-1}\\
=&\frac{n_p(\mathrm{p})}{d_p(\mathrm{p})}\frac{n_c(\mathrm{p})}{d_c(\mathrm{p})}\left(\frac{d_p(\mathrm{p})d_c(\mathrm{p})+n_p(\mathrm{p})n_c(\mathrm{p})}{d_p(\mathrm{p})d_c(\mathrm{p})}\right)^{-1}\\
=&\frac{n_p(\mathrm{p})}{d_p(\mathrm{p})}\frac{n_c(\mathrm{p})}{d_c(\mathrm{p})}\left(\!\frac{d_{pc}(\mathrm{p})}{d_p(\mathrm{p})d_c(\mathrm{p})}\!\right)^{-1}\!\!\!=\!n_p(\mathrm{p})n_c(\mathrm{p})\left(d_{pc}(\mathrm{p})\right)^{-1}\\
=&n_p(\mathrm{p})n_c(\mathrm{p})\frac{\overline{d_{pc}(\mathrm{p})}}{\overline{d_{pc}(\mathrm{p})}d_{pc}(\mathrm{p})}=\frac{n_p(\mathrm{p})n_c(\mathrm{p})\overline{d_{pc}(\mathrm{p})}}{\overline{d_{pc}(\mathrm{p})}d_{pc}(\mathrm{p})}\\
\end{aligned}
\end{equation}

 Assume that $d_{pc}(\mathrm{p})$ in (\ref{eq:proof2}) is generally written as
$d_{pc}(\mathrm{p})=a(\mathrm{p})e_0+b(\mathrm{p})e_1+c(\mathrm{p})e_2+d(\mathrm{p})e_{12}\in \mathcal{F}_{2,0}$, with $a(\mathrm{p}), b(\mathrm{p}), c(\mathrm{p}), d(\mathrm{p}) \in \mathcal{F}_{0,0}$. Since its geometric conjugate is given by $\overline{d_{pc}(\mathrm{p})}=a(\mathrm{p})e_0-b(\mathrm{p})e_1-c(\mathrm{p})e_2-d(\mathrm{p})e_{12}$~\cite{Hes84}, it follows that the denominator of  (\ref{eq:proof2}) is
\begin{equation}
d_{cl}(\mathrm{p})\!=\!\overline{d_{pc}(\mathrm{p})}d_{pc}(\mathrm{p})\!=\!(a(\mathrm{p})^2\!-\!b(\mathrm{p})^2\!-\!c(\mathrm{p})^2\!+\!d(\mathrm{p})^2)e_0
\end{equation}
is a real-valued polynomial function, $d_{cl}(\mathrm{p})\in \mathcal{F}_{0,0}$,  and its roots will determine closed-loop system stability.
\end{proof}

Hence, GA-valued closed-loop stability is analyzed using the same tools used in real-valued SISO systems: by studying the  roots of a real-valued polynomial denominator. Hence, well known tools like the  Routh-Hurwitz stability criterion for parametric settings or robust stability tests for dealing with uncertainties can be directly applied.


\begin{example}[Stability analysis with a proportional controller]\label{e:ex5}
Assume the following GA-TF-based proportional controller 
\begin{equation}\label{eq:pc1}
C_{\mathbb{G}}(\mathrm{p})=k(e_0+e_1), \quad k\in \mathbb{R}
\end{equation}
for the unbalanced geometric representation of the plant (GA-valued SISO system (\ref{eq:gsiso}) characterized by (\ref{eq:gmimoexunbalanced})) in closed-loop form as in Fig.~\ref{fig:gmimo}. The GA-valued closed-loop transfer function (\ref{eq:closedloop}) is given by 
\begin{equation}\label{eq:closedloopexs}
\begin{aligned}
G_{\mathbb{G}}^{cl}(\mathrm{p})&=( 3Rk(2R + (L+ Lu)\mathrm{p}) )d_{cl}(\mathrm{p})^{-1}(e_{0}+e_{1})\\
&+\sqrt{3}Rk(L - Lu)\mathrm{p}d_{cl}(\mathrm{p})^{-1}(e_2-e_{12})
\end{aligned}
\end{equation}
where 
\begin{equation}\label{eq:denclosedloopexs}
 \begin{aligned}
d_{cl}(\mathrm{p})=&2(L^2+2L_uL)\mathrm{p}^2+ 6R^2 + 12R^2k\\
&+4\left(2LR + LuR+ 2LRk + L_uRk\right)\mathrm{p}
\end{aligned}
\end{equation}
The roots of  real-valued polynomial (\ref{eq:denclosedloopexs}) with the numerical values for the components given in Example~\ref{e:ex1} and
$k=10$ are $s_1=-9.04\cdot 10^{4}$ and $s_2=-0.17\cdot 10^{4}$.
Hence, according to Proposition~\ref{p:stability}, the GA closed-loop system is stable. 
In fact, 
computing the location for the slowest  pole of the closed-loop system (\ref{eq:closedloopexs}) as a function of a wider range of values for the controller gain $k=10^{-6}\cdot10^{i}$, $i=0,1..12$. For any value, stability is guaranteed (the real part of the pole is always negative), and for small or high values of $k$, the slowest pole value collapses around $-1.04\cdot 10^3$.
\end{example}
%

\subsection{Migrating Youla--Kučera parametrization from RV to GA}\label{ss:youla}



The  $Q$ parameterization (or Youla--Kučera parameterization~\cite{You76}) of all stabilizing controllers is extended to the new GA domain. In the real-valued domain, the $Q$ parametrization states that for linear  closed-loop MIMO systems, where the real-valued plant $G(\mathrm{p})$ is a stable R-TF matrix, the family of all real-valued stabilizing  controllers is given by 
\begin{equation}\label{eq:youla}
C(\mathrm{p})=\left(I-Q(\mathrm{p})G(\mathrm{p})\right)^{-1}Q(\mathrm{p})
\end{equation}
where the parameter $Q(\mathrm{p})$ is a stable and proper R-TF matrix. 

\begin{proposition}\label{p:youla}
Assume that for a real-valued linear  closed-loop MIMO system, the family of all real-valued stabilizing negative feedback controllers $C(\mathrm{p})$ is given by (\ref{eq:youla}), where the plant $G(\mathrm{p})$ is stable, and the parameter $Q(\mathrm{p})$ is stable and proper. Then, for the GA-valued closed-loop scheme shown in Fig.~\ref{fig:gmimo}, the family of all GA-valued stabilizing negative feedback controllers $C_{\mathbb{G}}(\mathrm{p})$ is given by the GA-TF
\begin{equation}\label{eq:gayoula}
C_{\mathbb{G}}(\mathrm{p})=\left(e_0-Q_{\mathbb{G}}(\mathrm{p})G_{\mathbb{G}}(\mathrm{p})\right)^{-1}Q_{\mathbb{G}}(\mathrm{p})
\end{equation}
where $C_{\mathbb{G}}(\mathrm{p})=S^TT_{\mathbb{G}}C(\mathrm{p})T_{\mathbb{G}}S$, 
$G_{\mathbb{G}}(\mathrm{p})=S^TT_{\mathbb{G}}G(\mathrm{p})T_{\mathbb{G}}S$, and
$Q_{\mathbb{G}}(\mathrm{p})=S^TT_{\mathbb{G}}Q(\mathrm{p})T_{\mathbb{G}}S$, with $T_{\mathbb{G}}$  in (\ref{eq:transrg}) and $S$ in~(\ref{eq:selS}).
\end{proposition}

\begin{proof}
By applying the transformation (\ref{eq:transrg}) to (\ref{eq:youla}), the following expression is obtained
\begin{equation}\label{eq:youla1}
T_{\mathbb{G}}^{-1}C(\mathrm{p})T_{\mathbb{G}}\!=\!T_{\mathbb{G}}^{-1}\left(I-Q(\mathrm{p})G(\mathrm{p})\right)^{-1}Q(\mathrm{p})T_{\mathbb{G}}
\end{equation}
By using the property that the transformation $T_{\mathbb{G}}$~(\ref{eq:transrg})  is involutive, i.e., $T_{\mathbb{G}}=T_{\mathbb{G}}^{-1}$, expression (\ref{eq:youla1}) is  written as 
\begin{equation}\label{eq:youla2}
T_{\mathbb{G}}C(\mathrm{p})T_{\mathbb{G}}=T_{\mathbb{G}}\left(I-Q(\mathrm{p})T_{\mathbb{G}}T_{\mathbb{G}}G(\mathrm{p})\right)^{-1}T_{\mathbb{G}}T_{\mathbb{G}}Q(\mathrm{p})T_{\mathbb{G}}
\end{equation}
and  by selecting the first component using $S$ (\ref{eq:selS}) leads to
\begin{equation}\label{eq:youla21}
\begin{aligned}
&S^TT_{\mathbb{G}}C(\mathrm{p})T_{\mathbb{G}}S=\\
&S^TT_{\mathbb{G}}\left(I-Q(\mathrm{p})T_{\mathbb{G}}T_{\mathbb{G}}G(\mathrm{p})\right)^{-1}T_{\mathbb{G}}SS^TT_{\mathbb{G}}Q(\mathrm{p})T_{\mathbb{G}}S
\end{aligned}
\end{equation}
that rearranged is given by
\begin{equation}\label{eq:youla3}
\begin{aligned}
&\underbrace{S^TT_{\mathbb{G}}C(\mathrm{p})T_{\mathbb{G}}S}_{C_{\mathbb{G}}(\mathrm{p})}=\\
&(e_0-\underbrace{S^TT_{\mathbb{G}}Q(\mathrm{p})T_{\mathbb{G}}S}_{Q_\mathbb{G}}\underbrace{S^TT_{\mathbb{G}}G(\mathrm{p})T_{\mathbb{G}}S}_{G_{\mathbb{G}}})^{-1}\underbrace{TS^T_{\mathbb{G}}Q(\mathrm{p})T_{\mathbb{G}}S}_{Q_\mathbb{G}}
\end{aligned}
\end{equation}
which leads to (\ref{eq:gayoula}).
\end{proof}

Hence, in the GA domain, Proposition \ref{p:youla} indicates that for a given (open-loop stable) plant, the family of all stabilizing controllers can be computed using the same procedure used in the real domain, requiring the $Q$ parameter

\begin{equation}\label{eq:qplant}
Q_{\mathbb{G}}(\mathrm{p})=q_0(\mathrm{p})e_0+q_1(\mathrm{p})e_1+q_2(\mathrm{p})e_2+q_3(\mathrm{p})e_{12}
\end{equation}
to be stable.

\begin{example}[Decoupling GA-valued controller]\label{e:ex9} 
The  GA-TF controller design of the unbalanced scenario is now considered to target the decoupling effect. The design conditions for $Q_{\mathbb{G}}(\mathrm{p})$  that decouples the real-valued MIMO dynamics can be solved by inspection. Setting $q_0(\mathrm{p})$ as a free parameter, a decoupling controller may be found using the parameter whose coefficients are $q_1(\mathrm{p})\!=\!-g_1(\mathrm{p})q_0(\mathrm{p})g_0(\mathrm{p})^{-1}$,  $q_2(\mathrm{p})\!=\!-g_2(\mathrm{p})q_0(\mathrm{p})g_0(\mathrm{p})^{-1}$ and 
$q_3(\mathrm{p})\!=\!-g_3(\mathrm{p})q_0(\mathrm{p})g_0(\mathrm{p})^{-1}$, where $g_i(\mathrm{p})$ are the system's model coefficients.
Using  $q_0(\mathrm{p})=-1$, leads (\ref{eq:qplant}) to
\begin{equation}\label{eq:freeQ2}
\begin{aligned}
Q_{\mathbb{G}}(\mathrm{p})\!=&1e_0\!+\!(L\!-\!L_u)\mathrm{p}\left(e_1\!-\!\sqrt{3}e_2\right)d_q(\mathrm{p})^{-1}
\end{aligned}
\end{equation}
with $d_q(\mathrm{p})=(4L+2L_u)\mathrm{p}+6R$, which is stable and proper. By considering (\ref{eq:freeQ2}) , the parametrization (\ref{eq:gayoula}) leads to the GA controller 
\begin{equation}\label{eq:decontroller_sym}
C_{\mathbb{G}}(\mathrm{p})=\!\left(\!\left( 3R\!+\!L\right)\!e_0\!+\!(L\!-\!L_u)\mathrm{p}\!\left(\!e_0\!+\!e_1\!-\!\sqrt{3}e_2\!\right)\!\right)\!d_c(\mathrm{p})^{-1}
\end{equation}
with $d_c(\mathrm{p})=\left(2L + L_u\right)\mathrm{p} + 3R$. And replacing the system's values leads to
\begin{equation}\label{eq:decontroller}
C_{\mathbb{G}}(\mathrm{p})=-\left(1833\mathrm{p}^{-1}-1\right)e_0+0.375e_1-0.649e_2
\end{equation}
Fig.~\ref{fig:decoupling} displays two plots: one for the \(\alpha\) component and another for the \(\beta\) component of the input and output, respectively. The blue markers represent the desired setpoint at each time instant, while the red markers indicate the corresponding output values.
Observe that the set-point change introduced at  $t=0.05$ s only for the $\beta$ channel does not affect the $\alpha$ channel. The closed loop dynamics is driven by the expression
\begin{equation}\label{eq:close_dloop_dynamics}
    G^{cl}_{\mathbb{G}}=3R\left((2L + L_u)\mathrm{p} + 3R\right)^{-1}e_0
\end{equation}
which contains only the $e_0$  component, indicating that it has been completely decoupled. In fact, it is easy to double-check that the equivalent real-valued MIMO closed-loop transfer matrix is diagonal. Furthermore, due to the specific control methodology employed, the closed-loop response is as fast as the open-loop system dynamics. As a result, transients are extremely short and practically imperceptible in the experimental results.

\begin{figure}[t]
\centering
\includegraphics[width=6.5cm]{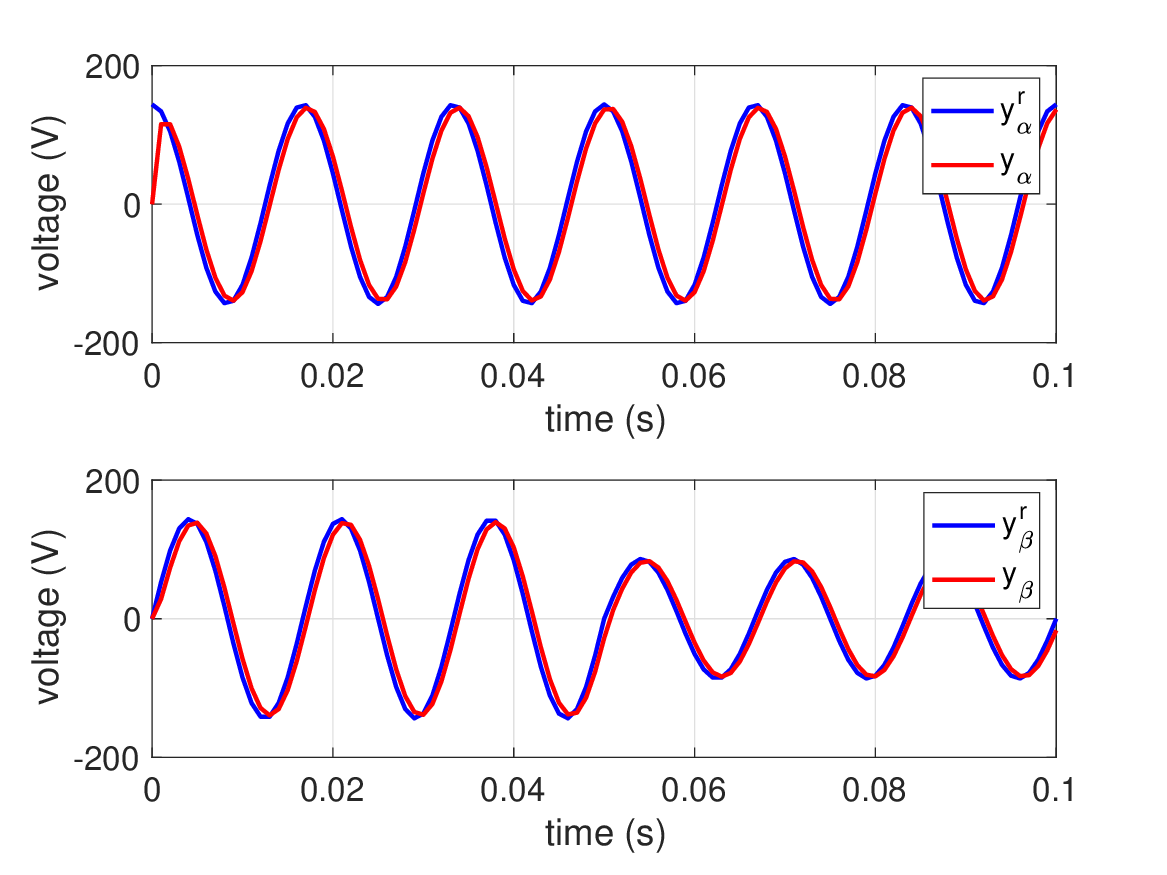}
\caption{Simulation results: Decoupling tracking controller (\(\alpha\) and \(\beta\) channels). The system model corresponds to the one depicted in Fig.~\ref{fig:scheme3}, using eq.~(\ref{eq:gmimoexunbalanced}). The controller is described by eq.~(\ref{eq:decontroller}). The inputs and outputs are transformed here into the \(\alpha\) and \(\beta\) channels.}
\label{fig:decoupling}
\end{figure}
\end{example}

\section{Experiments}\label{ss:res}

We have introduced a novel framework for three-phase systems representation and design that facilitates the handling of unbalanced systems with ease. This approach simplifies the complexity inherent in such systems and provides a structured methodology for their analysis and control. To further validate the practicality of this framework, we now proceed to deploy the proposed controller in a real system.

The three phase scheme shown in Fig.~\ref{fig:scheme3} has been reproduced in the laboratory to test the proposed controller (see Fig.~\ref{fig:lab}). The three ideal voltage sources are implemented using  a MTL-CBI0060F12IXHF GUASCH three-phase IGBT full-bridge power inverter with a rated power of $3.3$ kVA at $110$ V$_{rms}$ and $10$ A$_{rms}$ (central card in Fig.~\ref{fig:inv}),  connected though a $LC$ filter and a transformer. { The power inverter operates as a  switched power supply and the LC filter plus transformer are only used for filtering purposes in order to emulate the ideal  voltage sources. Hence, $v_a$, $v_b$ and $v_c$ in Fig.~\ref{fig:scheme3} corresponds to the line voltages found after the transformer in the  experimental setup.
The inductances $L_a$, $L_b$ and $L_c$ and loads $R$ given in Example~\ref{e:ex1} have been reproduced as follows. The load is composed by three heaters (one per line), and one line  has an additional serial inductance (Fig.~\ref{fig:load}) for creating the imbalance,  all reproducing the values given in Example~\ref{e:ex1}}.  The input of the  inverter is supplied by a Cinergia B2C+DC power source. The decoupling controller (\ref{eq:decontroller}) is implemented at the inverter on the F28M36 digital signal processor (DSPs) from Texas Instruments and executed every $T_s=100$~$\mu$s. The variables of interest, output voltages and currents, are measured and sent to a computer for its treatment and plotting.

The deployment of the controller to the F28M36 digital signal processor (DSP) is straightforward, achieved by applying the inverse of the variable transformation defined in eq.~(\ref{eq:transrg}). Subsequently, the controller is discretized using a pre-warping technique at 60 Hz. The resulting difference equations are then implemented as the controller's code. The implementation is carried out in ANSI C, using the compiler provided by Texas Instruments for this purpose and this DSP.

\begin{figure}[!t]
\centering
\subfloat[\small Inverter]{
\includegraphics[width=4.25cm]{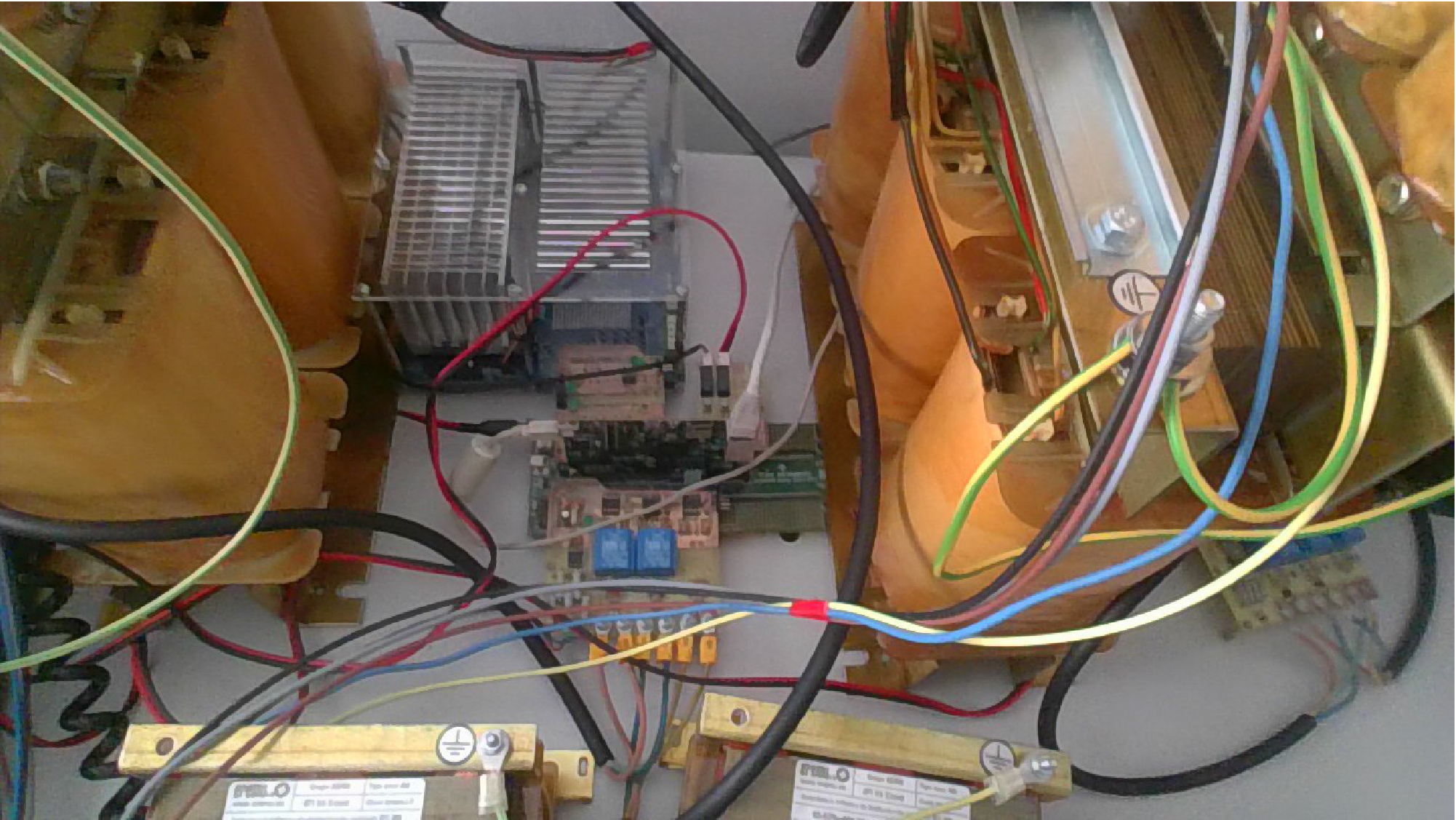}
\label{fig:inv}}
\subfloat[\small Unbalanced load]{
\includegraphics[width=4.25cm]{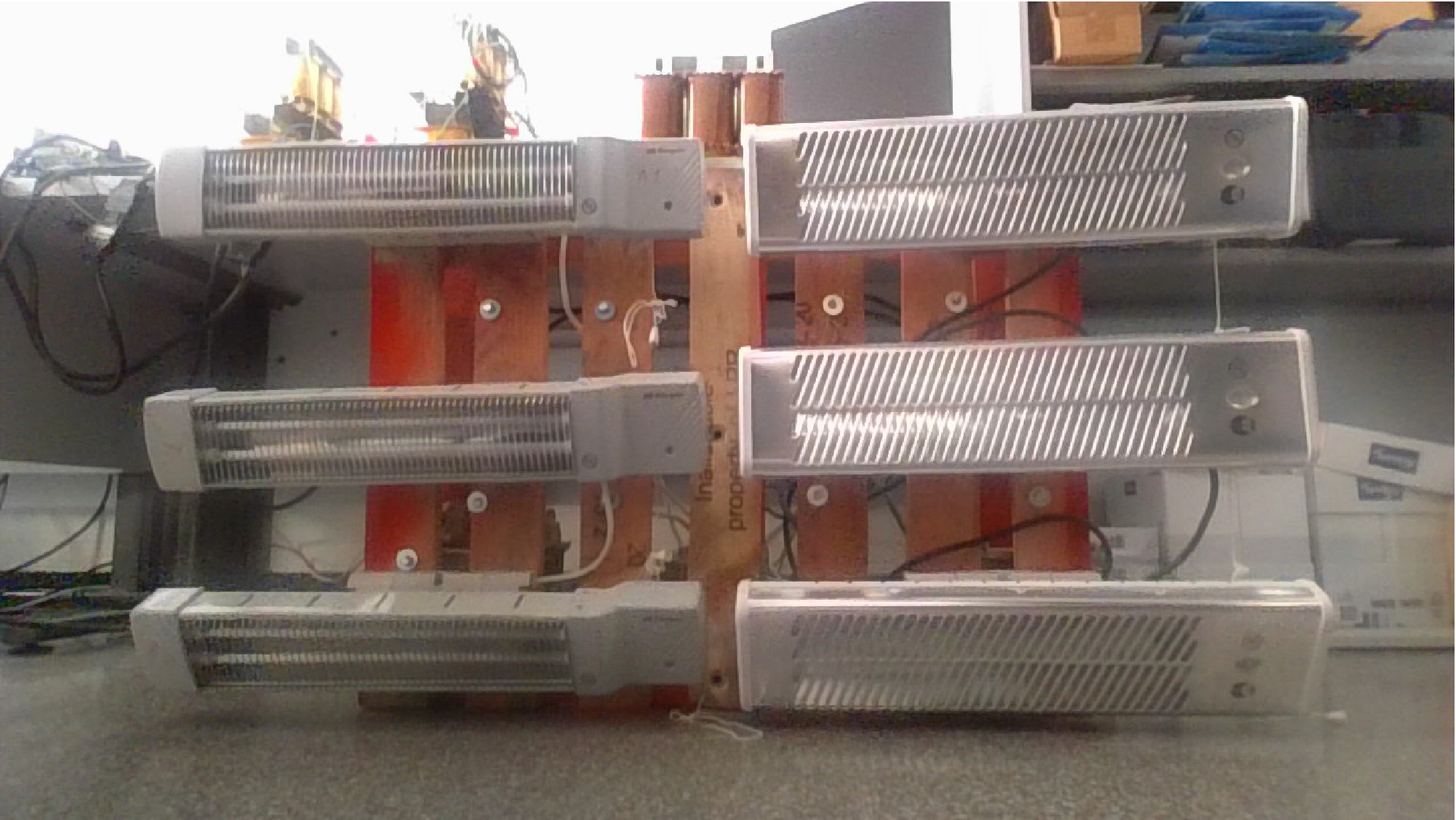}
\label{fig:load}}
\caption{Laboratory set-up.}
\label{fig:lab}
\end{figure}

\begin{figure}[t]
\centering
\subfloat[\small Open-loop]{
\includegraphics[width=7cm]{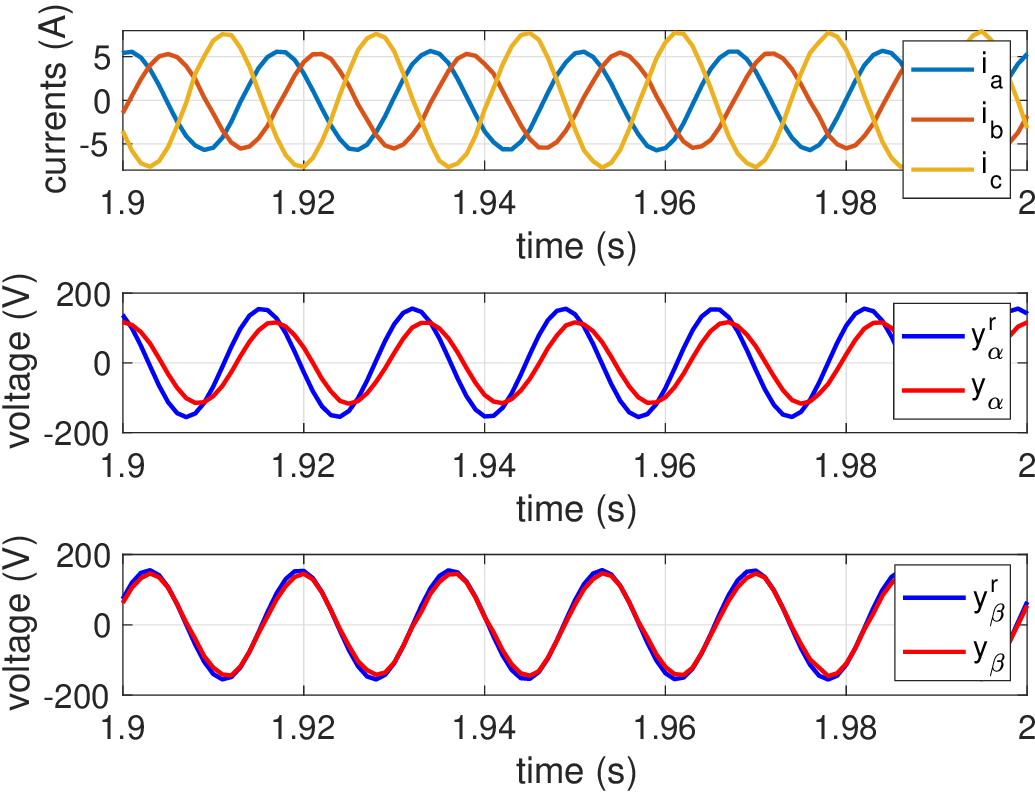}
\label{fig:ol}}
\vspace{0.3cm}
\subfloat[\small Closed-loop]{
\includegraphics[width=7cm]{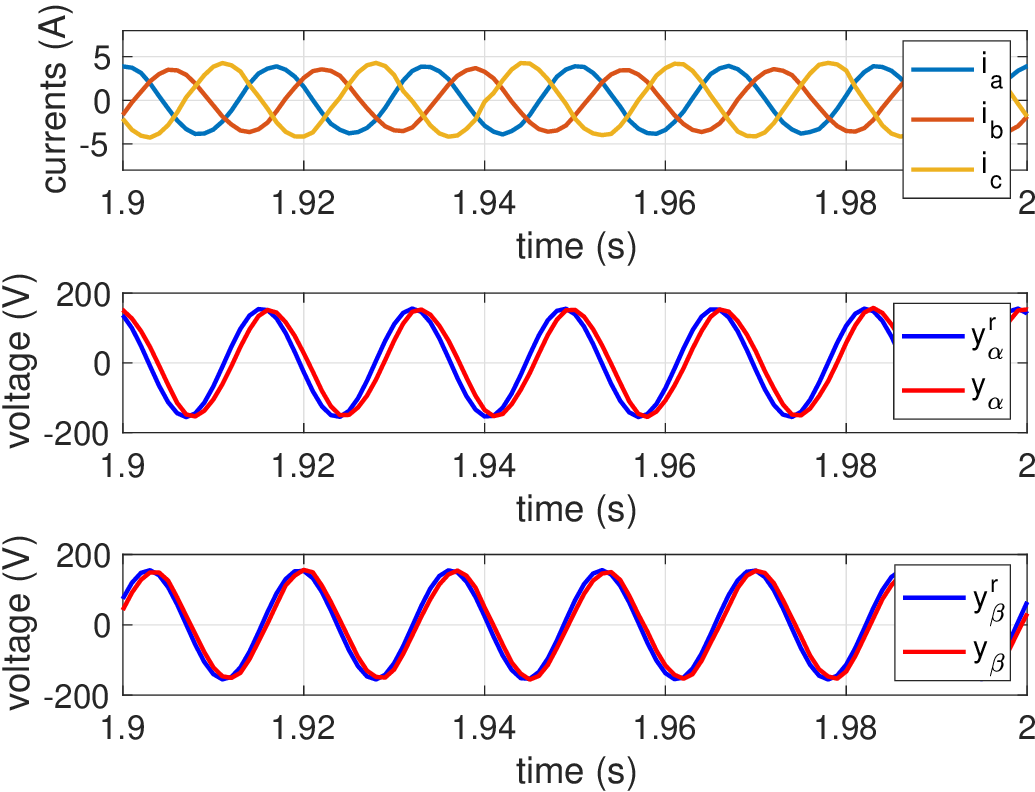}
\label{fig:cl}}
\caption{Experimental results: Geometric controller experiment.}
\label{fig:exp}
\end{figure}

Fig. \ref{fig:exp} shows the main results for the unbalanced case ($L_a=L_c$ and $L_b=L_u$), for both the open-loop (Fig.~\ref{fig:ol}) and closed-loop  (Fig.~\ref{fig:cl}) scenarios. For each scenario, the top plot shows  the measures of the three-phase output currents, $i^o_a$, $i^o_b$ and $i^o_c$, and the two plots below show the tracking performance (voltage reference  and output voltage for each $\alpha \beta$ channel, which are obtained from the output voltage measures  $v^o_a$, $v^o_b$ and $v^o_c$). For the open-loop scenario, Fig.~\ref{fig:ol}, the currents are clearly unbalanced, indicating an existing coupling, while voltage tracking performance is poor. The application of the decoupling controller (\ref{eq:decontroller}) achieves, apart from stable dynamics, balanced currents and decoupled voltage dynamics with the designed tracking accuracy. 

\section{Conclusions}\label{ss:con}
 Motivated by limitations that traditional modeling approaches have for unbalanced systems, this paper has introduced the use of GA for the dynamic modeling, analysis and controller design of three-phase electrical systems. The modeling approach, based on defining  a new GA mathematical framework, allows representing a balanced or an unbalanced three phase electrical system (which are multivariable systems) with a GA-valued linear SISO model whose plant is defined by a single GA-TF. Moreover, it has been shown that the stability analysis in the new GA domain simplifies to analyzing the roots of a real-valued polynomial, as it is also done  in the case of standard real-valued linear SISO models. Regarding the controller design phase, the Youla parametrization has been extended to the GA domain for the design of stable and decoupling controllers. 
 
The introduced GA approach for the analysis and control design of three-phase electrical systems  is a starting point for building a new GA-based systems theory. Therefore, it opens  a wide range of  new research directions that deserve being explored, from the electrical systems analysis point of view like reformulating a GA-based Ohm's law, or from the control systems perspective like extending frequency domain stability tools to the new  GA space (e.g. developing a new GA-based Nyquist criterion).  Moreover, presented results like the GA-stability condition should be updated to include robustness properties, and the GA decoupling  strategy calls for a comparison with existing MIMO paring/decoupling tools. The GA framework is demonstrated on an inverter with an unbalanced load for clarity and will be extended to larger unbalanced systems in future work.

\section*{Acknowledgments}

The work of I. Zaplana was partially supported by the Spanish National Project PID2020-114819GB-I00 funded by MICIU/AEI/10.13039/501100011033, and by the Generalitat de Catalunya through the Project 2021 SGR 00375. The work of A. Dòria-Cerezo was partially supported by the Spanish National Project PID2021-122821NB-I00 funded by MICIU/AEI/10.13039/501100011033 and ERDF/EU, and by the Generalitat de Catalunya through the Project 2021 SGR 00376. M. Velasco, J. Duarte and P. Martí were partially supported by the Spanish National Project PID2021-122835OB-C21 funded by MICIU/AEI/10.13039/501100011033.

{\appendices

\section*{Appendix: GA Basic Concepts}\label{ss:ap2}

In general, let $\mathbb{R}^{p+q}$ be a real vector space, where $p$ and $q$ are the number of basis vectors that square to $1$ and $-1$, respectively, i.e., the dimension of this real vector space is $n = p + q$. The associated GA
$\mathcal{G}_{p,q}(\mathbb{R})$ has $2^n$ basis elements, and the  objects of this algebra, called multivectors, are  linear combinations of them, where the coefficients belong to $\mathbb{R}$. The core idea of GA is its multiplication operation, called the geometric product 
~\cite{Hes84}. Every geometric algebra $\mathcal{G}_{p,q}(\mathbb{R})$ has as scalar basis element, which is denoted by $e_0$, and it plays the role of the identity for the geometric product.
If instead of $\mathbb{R}^{p+q}$, an arbitrary vector space over a field is considered, the associated GA is constructed in an analogous manner. 

\begin{observation}[Real numbers]
The GA representation of the real numbers space, $\mathbb{R}$, is given by $\mathcal{G}_{0,0}(\mathbb{R})$, or simply $\mathcal{G}_{0,0}$, where the only basis element is $e_0$. Hence, $\mathrm{a} \in \mathbb{R}$ can be represented as $\mathrm{a}e_0\in \mathcal{G}_{0,0}$. 
\end{observation}

\begin{observation}[Complex numbers]
The GA representation of the complex numbers space, $\mathbb{C}$, is given by $\mathcal{G}_{0,1}(\mathbb{R})$, or simply $\mathcal{G}_{0,1}$, where the only basis element besides $e_0$ is $e_1=j$ (that squares  $-1$). Hence, $\mathrm{a}+j\mathrm{b} \in \mathbb{C}$, with $\mathrm{a}, \mathrm{b}\in \mathbb{R}$, can be represented as $\mathrm{a}e_0+\mathrm{b}e_1\in \mathcal{F}_{0,1}$.
\end{observation}

\begin{observation}[Complex-valued transfer functions, C-TF]
The GA representation of the C-TF space is given by $\mathcal{G}_{0,1}(\mathbb{F})$, or $\mathcal{F}_{0,1}$ to distinguish it from the GA representation of the complex numbers. Its basis elements are $e_0$ and $e_1=j$. Hence, the C-TF $\mathrm{G_a(\mathrm{p})}+j\mathrm{G_b(\mathrm{p})}$ with $\mathrm{G_a(\mathrm{p})}, \mathrm{G_b(\mathrm{p})}\in \mathbb{F}$ can be represented as $\mathrm{G_a(\mathrm{p})}e_0+\mathrm{G_b(\mathrm{p})}e_1\in \mathcal{F}_{0,1}$.
\end{observation}

\begin{observation}[Geometric-valued transfer functions, GA-TF]
	The GA representation of the GA-TF space is given by $\mathcal{G}_{2,0}(\mathbb{F})$, or simply $\mathcal{F}_{2,0}$. Its basis elements are $e_0, e_1, e_2$ and $e_1e_2=e_{12}$, where $e_1e_2$ denotes the geometric product between vectors $e_1$ and $e_2$. Hence, a GA-TF can be represented as $\mathrm{G_a(\mathrm{p})}e_0+\mathrm{G_b(\mathrm{p})}e_1+\mathrm{G_c(\mathrm{p})}e_2+\mathrm{G_d(\mathrm{p})}e_{12}\in \mathcal{F}_{2,0}$ with $\mathrm{G_a(\mathrm{p})}, \mathrm{G_b(\mathrm{p})}, \mathrm{G_c(\mathrm{p})}, \mathrm{G_d(\mathrm{p})}\in \mathbb{F}$, i.e., R-TFs.
\end{observation}

}

 
%

\bibliographystyle{Bibliography/IEEEtranTIE}
\bibliography{Bibliography/IEEEabrv,Bibliography/FPNSC} 

\begin{IEEEbiography}[{\includegraphics[width=26mm, height=32mm, clip, keepaspectratio]{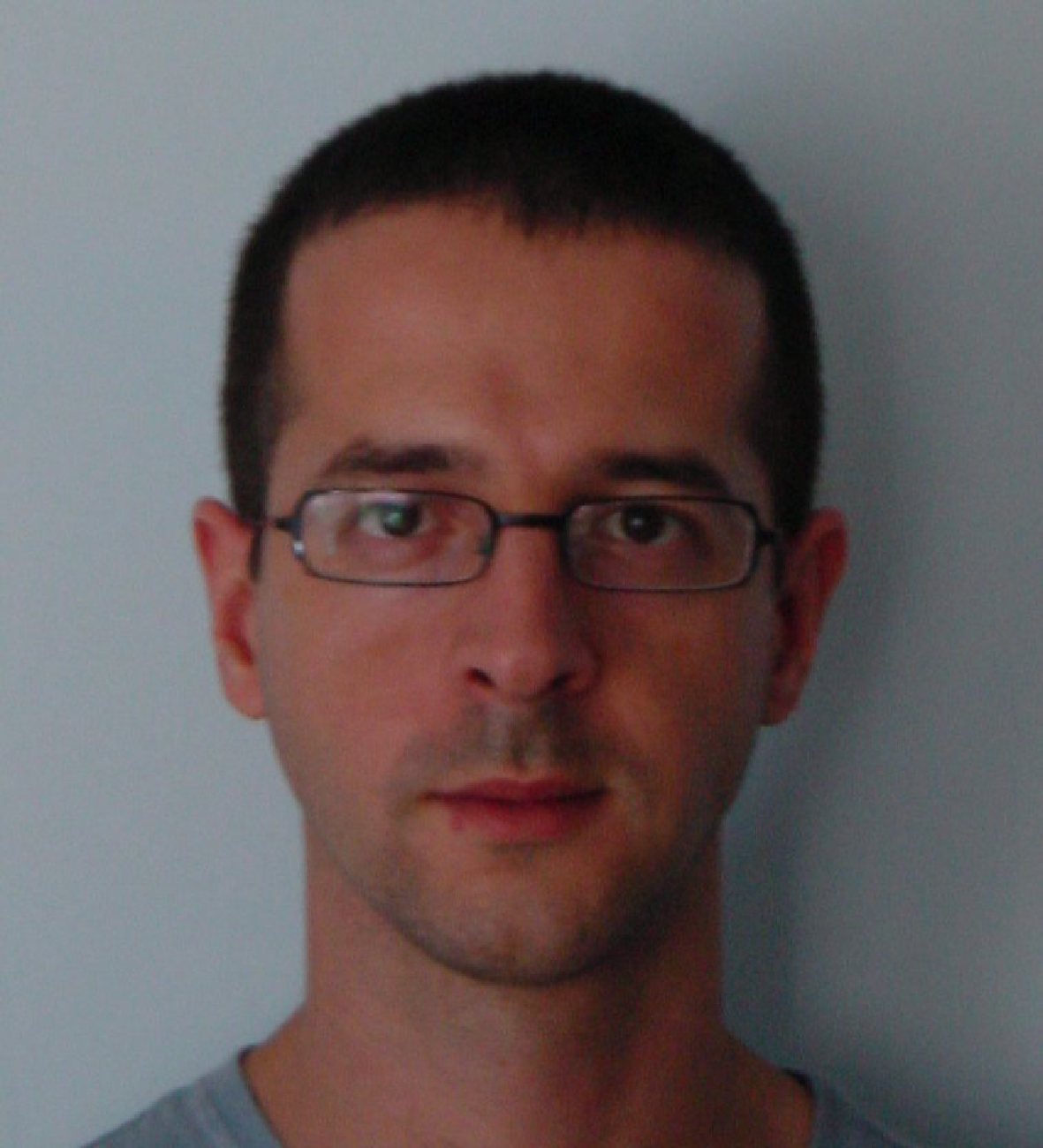}}]{Manel Velasco}
received the graduate degree in maritime engineering and the Ph.D. degree in automatic control from the Universitat Politècnica de Catalunya (UPC), Barcelona, Spain, in 1999 and 2006, respectively.
Since 2002, he has been an Assistant Professor with the Department of Automatic Control, Technical University of Catalonia. His research interests include artificial intelligence, real-time control systems, collaborative control systems, and microgrids.
\end{IEEEbiography}
\vspace{-1.25cm}
\begin{IEEEbiography}[{\includegraphics[width=26mm, height=32mm, clip, keepaspectratio]{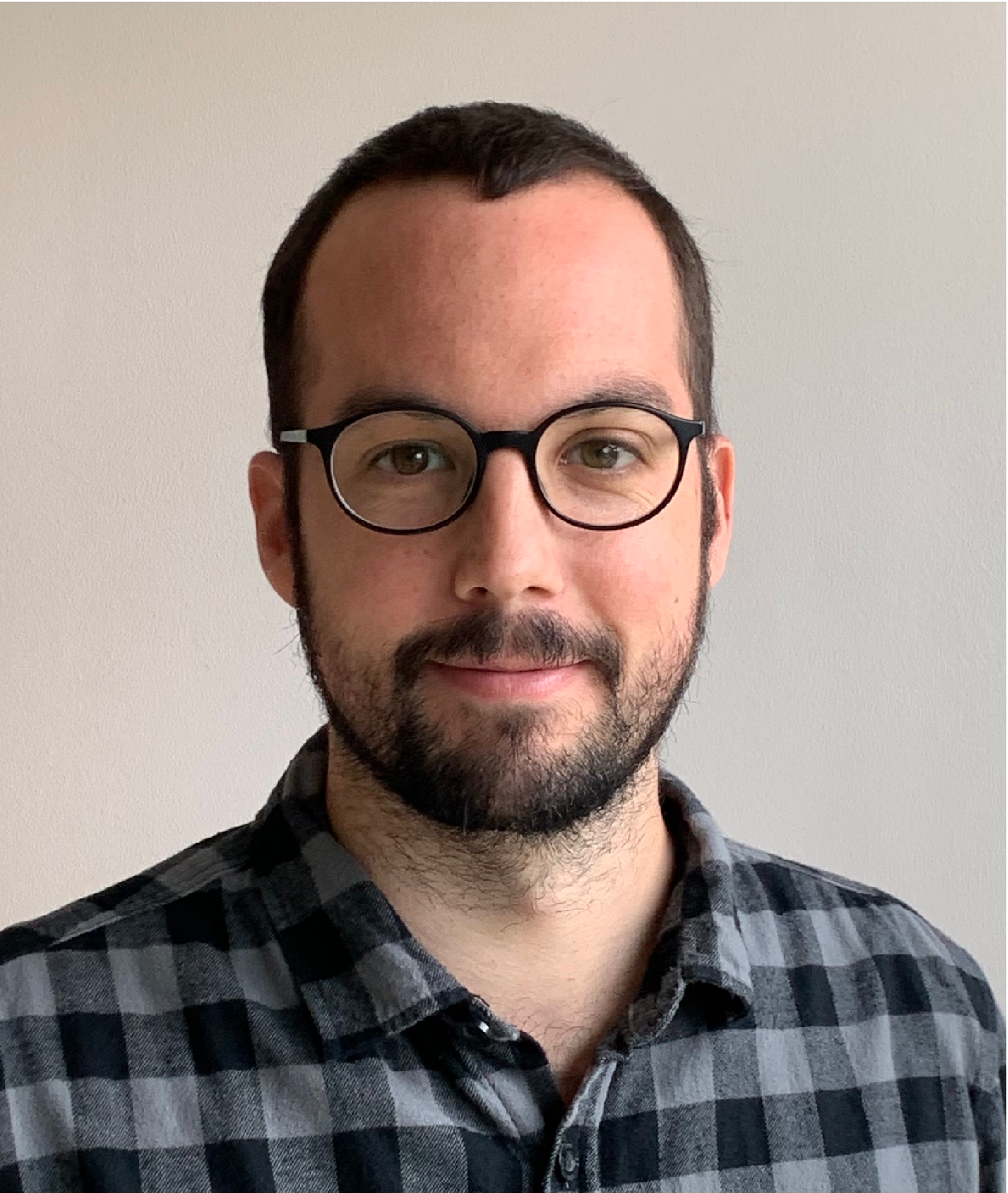}}]{Isiah Zaplana}
received a degree in mathematics and a Ph.D. degree in robotics, automatic control and vision from the Universitat Politècnica de Catalunya (UPC), Barcelona, Spain, in 2012 and 2018, respectively. He is currently serving as an assistant professor in robotics and computer vision at the same university. His research interests include the development of geometric, algebraic and numerical tools for robotics, vision and control.
\end{IEEEbiography}
\vspace{-1.25cm}
\begin{IEEEbiography}[{\includegraphics[width=26mm, height=32mm, clip, keepaspectratio]{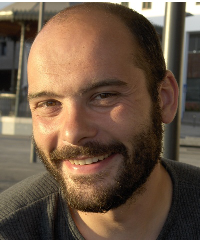}}]{Arnau Dòria-Cerezo}
received an undergraduate degree in electromechanical engineering from the Universitat Politècnica de Catalunya (UPC), Barcelona, Spain, in 2001, a DEA degree in industrial automation from the INSA-Lyon, France, in 2001, and a PhD degree in advanced automation and robotics from UPC in 2006. He is currently an Associate Professor with the Dept. of Electrical Engineering, UPC, and carries on his research activities at the Institute of Industrial and Control Engineering, UPC. 
\end{IEEEbiography}
\vspace{-1.25cm}
\begin{IEEEbiography}[{\includegraphics[width=26mm, height=32mm, clip, keepaspectratio]{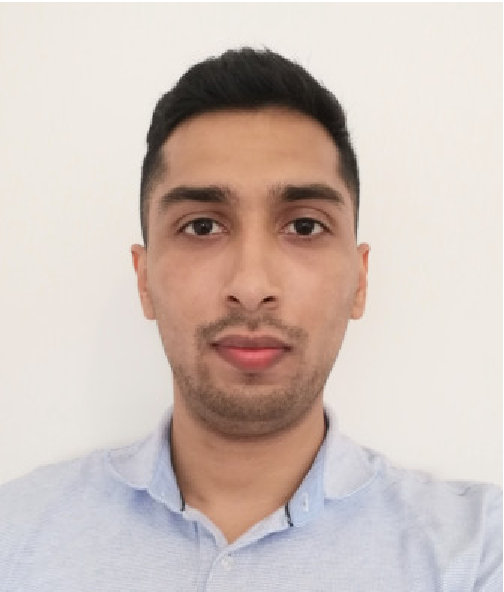}}]{Josué Duarte}
received the B.S. degree in mechatronic engineering from the Central American Technological University (UNITEC), Honduras, in 2018, and the Ph.D. degree in Automatic Control from the Universitat Politècnica de Catalunya (UPC), Barcelona, Spain, in 2024. He is currently an Assistant Professor in the Department of Electrical Engineering at the Escola Universitaria Salesiana de Sarria, Autonomous University of Barcelona, Barcelona, Spain. His research interests include control systems, microgrids, and power electronics.
\end{IEEEbiography}
\vspace{-1cm}
\begin{IEEEbiography}[{\includegraphics[width=26mm, height=32mm, clip, keepaspectratio]{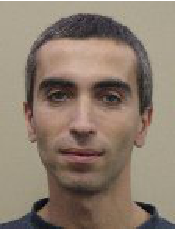}}]{Pau Martí}
received the degree in computer science and the Ph.D. degree in automatic control from the Universitat Politècnica de Catalunya (UPC), Barcelona, Spain, in 1996 and 2002, respectively. From 1996, he has held different research and teaching positions at the Department of Automatic Control at UPC. His research interests include embedded and networked control systems, nonlinear control and microgrids.
\end{IEEEbiography}

\end{document}